\documentclass[11pt,hyphens,dvipsnames]{article}
\usepackage[utf8]{inputenc} 
\usepackage[T1]{fontenc} 
\usepackage{xcolor}
\usepackage{hyperref} 
\hypersetup{
    colorlinks=true,
    linkcolor=blue,
    filecolor=magenta,      
    urlcolor=Magenta,
    citecolor=orange,
    pdftitle={Overleaf Example},
    }
\usepackage[round,sort&compress]{natbib}
\let\cite\citep
\usepackage{url}         
\usepackage{booktabs}
\usepackage{float}
\usepackage{amsfonts}    
\usepackage[margin=1.0in]{geometry}
\usepackage{amsmath,amssymb,amsfonts,mathtools,amsthm, thmtools}
\usepackage{url}
\usepackage{graphicx} 

\usepackage{subfig}
\usepackage{caption}
\newcommand{\mc}{\mathcal}
\newcommand{\mb}{\mathbb}

\DeclareMathOperator{\R}{\mathbb{R}}

\newcommand{\LFT}{{\tt LFT}}

\DeclareMathOperator{\spec}{\mathrm{spec}}
\newcommand{\cj}[1]{c_{#1}}
\newcommand{\cjc}[1]{c_{#1}^{\tt c}}
\newcommand{\xc}[1]{x_{#1}^{\tt c}}
\newcommand{\Lc}[1]{L_{#1}^{\tt c}}
\newcommand{\ellc}[1]{\ell_{#1}^{\tt c}}
\newcommand{\x}[1]{x_{#1}}
\newcommand{\Con}[1]{\mc{C}_{#1}}
\DeclareMathOperator*{\amin}{\mathrm{argmin}}

\newtheorem{assumption}{Assumption}
\newtheorem{theorem}{Theorem}
\newtheorem{lemma}{Lemma}

\newtheorem{corollary}{Corollary}
\newtheorem{definition}{Definition}
\newtheorem{proposition}{Proposition}
\newtheorem{remark}{Remark}
\newcommand{\bmat}[1]{\begin{bmatrix}
#1
\end{bmatrix}}

\newcommand{\edit}[1]{\textcolor{black}{#1}}
\title{ Consistent Conjectural Variations Equilibrium:\\
Characterization \& Stability for a Class of Continuous Games
}
\author{Daniel J. Calderone, Benjamin J. Chasnov, Samuel A. Burden, Lillian J. Ratliff}
\begin{document}
\maketitle

\begin{abstract}

Leveraging tools from the study of linear fractional transformations and algebraic Riccati equations, a local characterization of consistent conjectural variations equilibrium is given for two player games on continuous action spaces with costs approximated by quadratic functions. A discrete time dynamical system in the space of conjectures is derived,
a solution method for computing fixed points of these dynamics (equilibria) is given, local stability properties of the dynamics around the equilibria are characterized, and conditions are given that guarantee a unique stable equilibrium.
\end{abstract}

\section{Introduction}
In many multi-agent systems, the agents are learning about their opponents and the environment through interaction. Moreover, the agents often have bounded rationality---e.g., humans are known to not behave rationality~\cite{simon1955behavioral}, and machines inherently have bounded computational capabilities and are limited to making decisions based on their prescribed algorithmic process. 
Much of the literature on using game theory to model multi-agent systems has focused on static equilibrium notions that assume agents are rational such as Nash  or correlated equilibrium. These equilibrium concepts do not  capture the dynamic nature of learning systems or that in many cases agents form models of their opponent and react or optimize with respect to them. 

To address these issues, several different fields have examined the use of opponent models. The following examples are demonstrative.
In machine learning, opponent modeling~\cite{foerster2018learning,willi2022cola} can empirically improve the performance of reinforcement learning agents in some environments, while symplectic  methods can speed up convergence of gradient play in continuous games with certain structure~\cite{balduzzi2018mechanics,junsoo2022convergence}.
In game theory, opponent models known as \emph{conjectural variations}~\cite{figuieres2004theory} have been used to analyze strategic behaviors of firms in oligopoly and electricity markets~\cite{perry1982oligopoly,friedman2002bounded,liu2006empirical,diaz2010electricity}.
At the intersection of these areas, in prior work, we investigate the connection between gradient play and opponent anticipation leveraging conjectural variations~\cite{chasnov2020opponent}, and showed the relationship to implicit learning algorithms in  Stackelberg games~\cite{fiez2020implicit}.
Despite existing work there still remains several technical challenges in terms of characterizing the dynamic interaction of learning agents who form opponent models. 


Motivated by coupled non-cooperative learning systems wherein decision-makers have an opponent model and optimize with respect to this model, we provide a novel characterization of a (consistent) conjectural variations equilibrium ((C)CVE) \cite{bowley1924mathematical,frisch1933monopole}. A CVE is a non-cooperative equilibrium concept---predating even Nash---in which  each agent chooses their most favorable action taking into account that opponent strategies are a conjectured mapping of their own strategy. 
To gain intuition, a CVE can be thought of as a \emph{double sided Stackelberg equilibrium}. Indeed in a Stackelberg game, the leader best responds to a myopic follower---i.e. it solves $\min_x\{f(x,y)|\ y\in \amin_{y'}g(x,y')\}$. When both players act like a leader we have a double-sided Stackelberg game. 
This is a special case of a CVE wherein the conjecture is simply the myopic best response model of the follower. Conjectures can be more general mappings, however. Such an equilibrium is \emph{consistent} if each player's strategy in equilibrium is consistent with that which is conjectured by its opponent.
Unlike a Nash equilibrium, a (C)CVE handles strategic uncertainty through the use of conjectures, and has the following interpretation in terms of incentives: at a CVE no player has an incentive to deviate according to their own beliefs.
Our  interest in this equilibrium concept is precisely due to its aptitude for capturing  dynamic contexts, or situations of bounded (procedural)
rationality, or both. In particular, as we will highlight in the sequel, a (C)CVE can be seen as arising from repeated best response given an opponent model. 

\textbf{Contributions.} We leverage tools from the study of linear fractional transformations, and algebraic Riccati equations to provide a novel characterization of consistent conjectural variations equilibria for two-player $d_1\times d_2$ continuous games with quadratic costs; a quadratic game can also be thought of as an local approximation of more general costs. Focusing on conjectures that are affine in player actions, we derive a set of coupled Riccati equations and show that CCVE exist if these equations have solutions.
Additionally, we show that these coupled Riccati equations naturally lead to a discrete time dynamical system when they are iterated, ie. when players update their conjectures to be the best response to their opponent's conjecture in the affine conjecture space. 
We give a general solution method for computing fixed points of these dynamics (CCVE of the game) via solving an eigenvalue problem; we analyze the local stability properties of the dynamics around the CCVE; and we give conditions that guarantee a unique, stable CCVE. Finally, we discuss second order conditions and conclude with illustrative numerical examples and discussion.

\section{Preliminaries}
Consider the two-player game $\mc{G}=(f_1,f_2)$ such that $f_i\in C^2(\mb{R}^{d_1}\times \mb{R}^{d_2},\mb{R})$ for each $i\in \{1,2\}$. 
The function $f_i:\mb{R}^{d_1}\times \mb{R}^{d_2}\to \mb{R}$ is player $i$'s cost, which they seek to  minimize  by choosing   $x_i\in \mb{R}^{d_i}$. 
Let $x = (x_1,x_2) \in \mathbb{R}^d$ where $d=d_1 + d_2$ denotes the dimension of the joint action space.  
\edit{Let the set of \emph{conjectures} be the set of mappings defined by}
\[
\Con{1}\times\Con{2}=\{(\cj{1},\cj{2})|\ \cj{1}:\mathbb{R}^{d_2} \to \mathbb{R}^{d_1}, \ \cj{2}:\mathbb{R}^{d_1} \to \mathbb{R}^{d_2}\}
.\]

\begin{definition}\label{def:ccve}\rm
A point $\xc{} = (\xc{1},\xc{2})$ and a pair of conjectures $(\cjc{1},\cjc{2})\in \Con{1}\times\Con{2}$ is a \emph{consistent conjectural variations equilibrium (CCVE)} if  $\xc{i}=\cjc{i}(\xc{-i})$ for each $i=1,2$, and
\begin{equation*}
    \xc{i}=\amin_{x_i}\{f_i(x_i,x_{-i})|\ x_{-i}=\cjc{-i}(x_i)\},\quad\forall \ i=1,2.
\end{equation*}
\end{definition}
\edit{Given an a priori fixed set of conjectures $(\cjc{1},\cjc{2})\in \Con{1}\times \Con{2}$ in a CCVE, the point $ (\xc{1},\xc{2})$ is a generalized Nash equilibrium of the constrained game $\{\min_{x_i}f_i(x_i,\cjc{-i}(x_i))|\ x_i=\cjc{i}(\cjc{-i}(x_i))\}_{i=1}^2$. However,  finding a CCVE requires finding the maps $(\cjc{1},\cjc{2})$, so the problem of characterizing CCVE does not immediately reduce to a generalized Nash equilibrium problem.} 

As shown in \cite{bacsar1998dynamic}, when the costs are (jointly) strictly convex, an equivalent characterization of a CCVE in terms of the conjectures is the following: the two conjectures $(\cjc{1},\cjc{2})\in \Con{1}\times\Con{2}$ are a CCVE if and only if, for each $i=1,2$, we have
\begin{equation}\label{eq:ccvconj}
D_{x_i}f_i(x) + D_{x_{-i}}f_i(x) D_{x_i} \cjc{-i}(x_i) \equiv 0,\ \   x_i=\cjc{i}(x_{-i}),
\end{equation} 
where $D_{x}$ is the partial derivative operator with respect to a vector $x$.
\edit{In the absence of joint strict convexity, these are first order conditions; we call solutions to \eqref{eq:ccvconj} \emph{first-order CCVE}. A \emph{second-order CCVE} is a solution to \eqref{eq:ccvconj} with the additional condition that $\min_{x_i}\{f_i(x_i,\cjc{-i}(x_i))\}$ is strongly convex.} 

The focus of this paper is on characterizing CCVE and corresponding conjectures up to first- and second-order using a quadratic approximation of the game around the equilibrium. When the game is quadratic, a second-order CCVE is precisely a CCVE. 
Even in quadratic games, the existence of CCVE is not guaranteed, and as we show, for affine conjectures the question of existence boils down to finding solutions to coupled asymmetric Riccati equations. This is analogous to the existence of Nash equilibrium in dynamic linear quadratic games (cf.~\cite{aboukandil2003matrix}, \cite[Ch.~6]{bacsar1998dynamic}).
\subsection{Quadratic Game Approximation}
\label{subsec:quadratic-game-approximation}

The local quadratic approximation of cost $f_i$ is given by
\begin{align*}
\label{eq:cost_form}
    f_i(x_i,x_{-i})&=\frac{1}{2}\bmat{x_i\\ x_{-i}}^\top\bmat{A_i & B_i^\top\\ B_i& D_i}\bmat{x_i\\ x_{-i}}+\bmat{a_i\\ b_{i}}^\top\bmat{x_i\\ x_{-i}},
\end{align*}
where $A_i\in \mb{R}^{d_i\times d_i}$, $D_i\in\mb{R}^{d_{-i}\times d_{-i}}$, $B_i\in \mb{R}^{d_{-i}\times d_i}$, $a_i\in \mb{R}^{d_i}$ and $b_i\in \mb{R}^{d_{-i}}$. Further, we assume that $A_i\succ 0$ for each $i=1,2$.
The $D_i$ matrices penalize player $i$ based solely on $x_{-i}$ and may often be negative or zero.
As noted quadratic games are a useful approximation of the behavior of more complex games around an equilibrium. 
Moreover,  quadratic games of the form considered capture finite time linear quadratic games with open-loop strategies, since the dynamics can effectively be ``unrolled" and the strategy $x_i$ is simply the stacked vector of control inputs.

We consider only the space of affine conjectures; analogous to affine optimal policies in linear quadratic optimization problems, affine conjectures are the most natural class of conjectures for quadratic games as will be illustrated through our analysis. In fact, it is straightforward to show that if a player has an affine conjecture for its opponent, then the best response for that player is itself an affine policy. 
With this in mind, let 
player $i$ have an affine conjecture given by
$
x_{-i}  = \cj{-i}(x_i) =  L_{i}x_i + \ell_{i}$.
This results in player $i$ facing the following optimization problem:
\[\min_{x_i}\{f_i(x_i,x_{-i})|\ x_{-i}=c_{-i}(x_i) = L_{i}x_i+\ell_{i}\}.\]
The conditions for a first-order CCVE $(x_1,x_2)$ in affine conjectures are 
\begin{equation}\label{eq:firstorderone}
  \begin{aligned}
0&=D_{x_1}f_1(\x{1},\cj{2}(\x{1})), \  0=D_{x_2}f_2(\cj{1}(\x{2}),\x{2}),\\
\cj{2}(x_1) &=L_1\x{1}+\ell_1, \quad\ \ \cj{1}(x_2)=L_2\x{2}+\ell_2.
\end{aligned}  
\end{equation}
Given \eqref{eq:firstorderone}, the implications for existence can be summarized in the following proposition.
\begin{proposition}
For a quadratic game $(f_1,f_2)$, given affine conjectures of the form 
$\cj{-i}(x_{i})=L_{i}x_{i}+\ell_{i}$
for $i=1,2$, a first-order CCVE exists if there are solutions to coupled Riccati equations:
 \begin{align}
L_{-i}^\top  (A_i + B_i^\top  L_i)  +(B_i +D_iL_i)=0, \quad \forall\   i\in \{1,2\}.     \label{eq:coupledriccati}
 \end{align} 
 In addition, if a solution to \eqref{eq:coupledriccati} satisfies
  \begin{align}
A_i + L_i^\top B_i + B_i^\top L_i + L_i^\top D_i L_i  \succ 0
\quad \forall \ i\in \{1,2\},
\label{eq:2ndorderprop}
\end{align}
then that solution is a CCVE.
%
\end{proposition}

Indeed, applying the chain rule, player $i$'s first order optimality conditions are 
\begin{align*}
0 & = 
x_i^\top ( A_i +  B_i^\top L_i ) + x_{-i}^\top  ( B_i+D_iL_i) + a_i^\top + b_i^\top L_{i}. 
\end{align*}
If $x_i$ is consistent with player $-i$'s conjecture, then $x_i = L_{-i}x_{-i} + \ell_{-i}$. Plugging this in for $x_{i}$ gives 
\begin{align*}
0&=x_{-i}^\top (L_{-i}^\top (A_i + B_i^\top L_i) + B_i + D_iL_i)\\
&\quad+ \ell_{-i}^\top (A_i + B_i^\top L_i) + a_i^\top + b_i^\top L_i.
\end{align*}
For this to be true for any $x_{-i}$, we need the conditions in \eqref{eq:coupledriccati} to hold.
Equivalently, (and perhaps more intuitively) supposing the inverse of $(A_i+B_i^\top L_i)$ exists, we can rewrite player $i$'s first order condition as $ x_i^\top  = -x_{-i}^\top (B_i +D_iL_i)(A_i + B_i^\top  L_i)^{-1} - (a_i^\top  + b_i^\top  L_i)(A_i + B_i^\top  L_i)^{-1} $  and the consistent conjecture conditions are given directly by 
 \begin{equation}
 \label{eq:L2fromL1} 
     \begin{aligned}
     L_{-i}^\top  & = -(B_i +D_iL_i)(A_i + B_i^\top  L_i)^{-1},\\
      \ell_{-i}^\top  & = - (a_i^\top  + b_i^\top  L_i)(A_i + B_i^\top  L_i)^{-1} \ \ \ \forall\ i\in \{1,2\}.
     \end{aligned}
 \end{equation}
This shows that if a player has an affine conjecture for its opponent's play, then its best response can be written as an affine policy. Note that the Riccati equations \eqref{eq:coupledriccati} are sufficient for first order conditions since $\ell_i$ can be computed separately based on $L_i$ for $i=1,2$.

For quadratic games, a second-order CCVE is equivalent to a CCVE. 
The following conditions characterize when a second-order CCVE exits.
Expanding out player $i$'s cost given the affine conjecture, we get
\begin{align*}
f_i(x_i,\cj{-i}(x_i)) 
& =
\tfrac{1}{2}x_i^\top (A_i + L_i^\top B_i + B_i^\top L_i + L_i^\top D_i L_i ) x_i \\
&
+ (a_i^\top + \ell_i^\top B_i +b_i^\top L_i ) x_i
+ \ell_i^\top D_i \ell_i + b_i^\top \ell_i.
\end{align*}
Hence, player $i$'s optimization problem is strongly convex if \eqref{eq:2ndorderprop} hold.
We use $(\Lc{i},\ellc{i})$ to refer consistent conjectures---i.e., the solutions to the coupled Riccati equations \eqref{eq:coupledriccati} and the corresponding affine offsets. 
Solutions may still exist when the inverses in \eqref{eq:L2fromL1} do not, however, as has been shown in special cases in the literature on CCVE such as for scalar Bertrand games, this leads to a multiplicity of solutions and an equilibrium selection problem (see \cite{olsder1981memo, figuieres2004theory} and references therein). Given page constraints, we leave the analysis of these more nuanced cases to a future paper. 

For each $i=1,2$, define the following linear fractional transformation (LFT) update: 
\begin{equation*}
L^+_{-i}  = \LFT_{i,-i}( L_i ) = -(A_i^\top  + L_i^\top  B_i )^{-1}(B_i^\top  +L_i^\top D_i^\top ), 
\end{equation*}
where the subscript $(\cdot)_{12}$ can be read as ``from 1 to 2". The update for $L_i$ naturally defines discrete-time dynamics in the conjecture parameter space 
that show how a player should update their conjecture to be consistent with their opponent's current conjecture. 
It is also useful to think of dynamic updates for each player separately constructed by composing the updates as follows:
\begin{align}
L^+_i  & = \LFT_{-i,i}(
\LFT_{i,-i}(L_i)
) \label{eq:update1} \\
& = - \left(A_{-i}^\top - (B_i +D_iL_i)(A_i + B_i^\top  L_i)^{-1} B_{-i}\right)^{-1} \notag \\
&\quad\cdot(B_{-i}^\top  - (B_i +D_iL_i)(A_i + B_i^\top  L_i)^{-1} D_{-i}^\top ), \ \ i=1,2. \notag
\end{align}
\begin{remark}
The first order conditions in \eqref{eq:coupledriccati}
guarantee that the players have consistent conjectures. The second order conditions \eqref{eq:2ndorderprop} guarantee that given their conjecture, player $i$'s cost is 
convex in $x_i$. Expounding the first order conditions---characterizing the LFT dynamics, finding fixed points by solving  \eqref{eq:coupledriccati}, and characterizing their stability---is non-trivial and is the primary focus of this paper. Our results will show that there is a limited number of stable first-order CCVE.  Once these stable equilibria are found, the second order conditions \eqref{eq:2ndorderprop} can easily be checked.  For further discussion, see Section \ref{sec:2ndorder}.
\end{remark}



\subsection{LFT Matrix Representation}
We will see in the subsequent section that LFTs can be efficiently represented by matrices and their composition by matrix manipulation. Towards this end, let us define some useful objects that will be used throughout.
Define the $d\times d$ real valued matrices (where $d = d_1 + d_2$)
\begin{equation}\label{eq:LFTmatrices}
M_1  = 
\begin{bmatrix}
A_1 & B_1^\top  \\
B_1 & D_1
\end{bmatrix}, \ \ \text{and}\ \ 
M_2 = 
\begin{bmatrix}
D_2 & B_2 \\
B_2^\top  & A_2
\end{bmatrix}.
\end{equation}
We make the following assumption on $M_1$ and $M_2$. 
\begin{assumption}\label{a:well-defined} The matrices $M_1,M_2$ are invertible.
\end{assumption}
We will be directly interested in the two products $\mathbf{M}_1 = M_2^{-\top} M_1$ and $\mathbf{M}_2 = M_1^{-\top} M_2$.
Note that 
$
M_1,M_2 \ \text{invertible}
\iff 
\mathbf{M}_1,\mathbf{M}_2 \ \text{invertible} 
$
Let $\text{spec}\big(\mathbf{M}_1\big)$ and $\text{spec}\big(\mathbf{M}_2\big)$ refer to the spectra of each matrix.  A simple argument shows that $\text{spec}(\mathbf{M}_1) = 1/\text{spec}(\mathbf{M}_2)$ where we use $1/(\cdot)$ to mean element-wise inversion.  

\subsection{Examples}
In this section, we present two examples of consistent conjectural variations equilibria in quadratic games. 
\subsubsection{Linear quadratic dynamic game}
Consider a two player linear quadratic dynamic game with open loop policies $\mathbf{u}_i=(u_{i,0},\ldots, u_{i,T-1})$ for $i=1,2$:
\begin{align*}
   f_i(\mathbf{u}_1,\mathbf{u}_2)&=\sum_{t=0}^{T-1} \ \ \frac{1}{2}z_t^\top Q_iz_t+\frac{1}{2} {u}_{i,t}^\top R_i {u}_{i,t}+{u}_{i,t}^\top R_{i,-i} {u}_{-i,t}+\frac{1}{2} z_{T}^\top Q_{i,f}z_T\\ 
   z_{t+1}&=Fz_t+G_1u_{1,t}+G_2u_{2,t}, \ z_t\in \mb{R}^{n}.
\end{align*}
%
%
Unfolding the dynamics we have that for $Z=[z_0^\top,\ldots,z_T^\top]^\top$, we have $Z=W_1\mathbf{u}_1+W_2\mathbf{u}_2+\bar{F}z_0$
where \begin{align*}
W_i&=\bmat{0& \cdots& & & 0\\ G_i& 0 &\cdots& &0\\
  FG_i & G_i& 0& \cdots&0\\
    \vdots &\vdots &\ddots &\ddots&\vdots \\
    F^{T-2}G_i& F^{T-3}G_i& \cdots & G_i & 0\\
    F^{T-1}G_i& F^{T-2}G_i& \cdots & FG_i& G_i},\quad i=1,2,
\end{align*}
and \[\bar{F}=\bmat{I & F^\top  \cdots (F^{T})^\top }^\top.\]
Define the following cost matrices:
\begin{align*}
    \mathbf{Q}_i&:=\text{diag}(Q_i,\ldots, Q_i,Q_{i,f})\in \mathbb{R}^{n(T+1)\times n(T+1)},\\
    \mathbf{R}_i&:=\text{diag}(R_i,\ldots, R_i)\in \mb{R}^{d_iT\times d_iT},\\
 \mathbf{R}_{i,-i}&:=\text{diag}(R_{i,-i},\ldots, R_{i,-i})\in \mb{R}^{d_iT\times d_{-i}T}.
\end{align*} 
Hence player $i$'s cost is given by
\begin{align*}
    f_i(&\mathbf{u}_i,\mathbf{u}_{-i})= \frac{1}{2} \mathbf{u}_i^\top \mathbf{R}_i\mathbf{u}_i+\mathbf{u}_{i}^\top \mathbf{R}_{i,-i}\mathbf{u}_{-i}
    +\frac{1}{2} (W_1\mathbf{u}_1+W_2\mathbf{u}_2+\bar{F}z_0)^\top\mathbf{Q}_i (W_1\mathbf{u}_1+W_2\mathbf{u}_2+\bar{F}z_0).
\end{align*}
Expanding and regrouping this cost, we can recover $(A_i,B_i,D_i,a_i,b_i)$ for each player. Indeed, we have that
    \begin{align*}
    f_i(\mathbf{u}_1,\mathbf{u}_2)
    &= \tfrac{1}{2} \mathbf{u}_i^\top (\mathbf{R}_i+W_i^\top \mathbf{Q}_iW_i)\mathbf{u}_i+ \tfrac{1}{2} z_0^\top \bar{F}^\top \mathbf{Q}_i\bar{F}z_0+\mathbf{u}_i^\top(\mathbf{R}_{i,-i}+W_i^\top \mathbf{Q}_iW_{-i})\mathbf{u}_{-i}\\
    &\quad+z_0^\top \bar{F}^\top \mathbf{Q}_i(W_i\mathbf{u}_i+W_{-i}\mathbf{u}_{-i}) + \tfrac{1}{2} \mathbf{u}_{-i}^\top W_{-i}^\top \mathbf{Q}_iW_{-i}\mathbf{u}_{-i},\\
    &=\frac{1}{2}\bmat{\mathbf{u}_i\\ \mathbf{u}_{-i}}^\top\bmat{(\mathbf{R}_i+W_i^\top \mathbf{Q}_iW_i) &(\mathbf{R}_{i,-i}+ W_i^\top \mathbf{Q}_iW_{-i}) \\(\mathbf{R}_{i,-i}+ W_i^\top \mathbf{Q}_iW_{-i})^\top & W_{-i}^\top \mathbf{Q}_iW_{-i} }\bmat{\mathbf{u}_i\\ \mathbf{u}_{-i}}\\
    &\quad+z_0^\top \bar{F}^\top \mathbf{Q}_i(W_i\mathbf{u}_i+W_{-i}\mathbf{u}_{-i})+\tfrac{1}{2} z_0^\top \bar{F}^\top \mathbf{Q}_i\bar{F}z_0,
\end{align*}
so that 
\begin{align*}
A_i & =\mathbf{R}_i+W_i^\top \mathbf{Q}_iW_i, \qquad \qquad \qquad  \ \ 
a_i^\top  =z_0^\top \bar{F}^\top \mathbf{Q}_iW_i \\
B_i & =\big(\mathbf{R}_{i,-i}+W_i^\top \mathbf{Q}_iW_{-i}\big)^\top  \qquad \qquad 
b_i^\top =z_0^\top \bar{F}^\top \mathbf{Q}_iW_{-i} \\
D_i & = W_{-i}^\top \mathbf{Q}_iW_{-i} 
\end{align*}
In a typical LQR problem it is assumed that $\mathbf{R}_i\succ 0$ in order for solutions to exist (there are conditions that weaken this assumption), and hence $A_i\succ 0$. Since $A_i$ is non-degenerate under the assumption $\mathbf{R}_i\succ 0$, a sufficient condition for $\mathbf{M}_i$ for $i=1,2$ to each be non-degenerate is that the Schur complement of $M_i$ with respect to $(\mathbf{R}_i+W_i^\top Q_iW_i)$ is non-degenerate; indeed, this follows from the fact that \[[\det(M_i)\neq 0\ \forall i\in\{1,2\}]\Longleftrightarrow[\det(\mathbf{M}_i)\neq 0\ \forall i\in\{1,2\}].\]

\subsubsection{Adaptive human-machine interactions}
It has recently been shown that CCVE well-model human-machine co-adaptation~\cite{chasnov2023human}. In this study the human and the machine have scalar quadratic costs, and series of experiments explore convergence of repeated game play to CCVE in an computer-facilitated task. The costs for the human and the machine are given by
\begin{align*}
    f_i(x_i,x_{-i})=\frac{1}{2}\bmat{x_i\\ x_{-i}}^\top\bmat{q_i & r_i\\ r_i& s_i}\bmat{x_i\\ x_{-i}}+\bmat{w_i\\ v_{i}}^\top\bmat{x_i\\ x_{-i}},
\end{align*}
where all the cost parameters are scalars and $x_i\in \mb{R}$ for each $i=1,2$. Assumption~\ref{a:well-defined} is satisfied for this game if $\det(M_i)\neq 0$ $\Longleftrightarrow$ $q_is_i-r_i^2\neq 0$ for each $i=1,2$. This holds for the games studied in \cite{chasnov2023human}, and further it is shown in the supplement of the same reference that CCVE exist in affine conjectures for the games studies therein.

\subsection{Warm-Up: Scalar M\"obius Transformations}
In order to get some intuition for why the matrices $\mathbf{M}_i$ have the form they, it is instructive to consider the scalar setting. 
It turns out that for scalar games the LFT description of the CCVE conjecture parameters is equivalent to a M\"obius transformation, and examining this case provides useful intuition for the more general case.

Indeed, the variables $L_1, L_2$ and parameters are all now scalar so that the LFT for player one reduces to 
\begin{equation*}
    L_1 = \LFT_{21}(L_2) = -(B_2  +L_2  D_2)/(A_2  + L_2B_2),
\end{equation*}
with composition map $\LFT_{21}\circ \LFT_{12}(L_1)$ given by
\begin{equation*}
   L_1  = 
\frac{(B_1D_2-A_1B_2) + (D_1D_2-B_1B_2)L_1}{(A_1A_2 - B_1 B_2) + (A_2 B_1 - D_1B_2)L_1}  
\end{equation*}
The expressions for player two are analogous.
It is well known that a M\"obius transformation can be represented in the matrix form
given in \eqref{eq:LFTmatrices} with scalar the entries \cite[Ch. VII]{lang2013complex}.
%
The matrix representation of the composition map is
\begin{align*}
\boldsymbol{\mathrm{M}}_1
&= 
\begin{bmatrix}
A_1 A_2 - B_1 B_2 & A_2 B_1 - D_1 B_2 \\
B_1 D_2 -A_1B_2& D_1 D_2 -B_1B_2
\end{bmatrix}\\
&=
\begin{bmatrix}
A_2 & -B_2 \\
-B_2 & D_2 
\end{bmatrix}
\begin{bmatrix}
A_1 & B_1 \\ 
B_1 & D_1 
\end{bmatrix}=\text{det}(M_2)M_2^{-1}M_1
\end{align*}
Note that scaling by the determinant is unimportant since the the matrix representation does not change with scaling. 

Fixed points of M\"obius transformations can be characterized in terms of the eigenvectors of their matrix representations. 
For a scalar LFT with matrix representation $\mathbf{M}_1$ there are two fixed points each characterized by an eigenvector of $\mathbf{M}_1$. 
Specifically, suppose {\small $[1 \ v]^\top$} is a right eigenvector of $\mathbf{M}_1=[m_{11} \ m_{12}; m_{21}\ m_{22}]$. Then we have 
\begin{align*}
\begin{bmatrix} m_{11} & m_{12} \\ m_{21} & m_{22} \end{bmatrix} \begin{bmatrix} 1 \\ v \end{bmatrix} = 
\begin{bmatrix} m_{11} + m_{12}v \\ m_{21} + m_{22} v \end{bmatrix} =
\begin{bmatrix} 1 \\ v \end{bmatrix} \lambda
\end{align*}
and it follows that $v = (m_{21} + m_{22} v)/( m_{11} + m_{12}v)$, ie. $v$ is a fixed point for the LFT.  Further analysis shows that the stability of each fixed point depends on the ratio of the two eigenvalues with one stable and one unstable fixed point or two marginally stable fixed points \cite{lang2013complex}.  

In order to extend this analysis to matrix LFTs, we use tools from algebraic Riccati equation analysis. Yet, at its core the idea is very much the same as in the scalar case in that the eigenstructure of $\mathbf{M}_1$ tells us everything about the stability properties of CCVE as we will see in the coming sections.

\section{LFT Dynamics: Matrix Form}
In this section, we study the composite LFT dynamics; fixed points of these dynamics define the conjecture parameters $(L_1,L_2)$ in a CCVE $(\x{1},\x{2})$. Recall that given $(L_1,L_2)$, the affine terms $(\ell_1,\ell_2)$ follow immediately from \eqref{eq:L2fromL1}. Given $(L_i, \ell_i)$ for $i=1,2$, we can easily recover $(\x{1},\x{2})$ by solving the linear equations $\{\x{-i}=L_i\x{i}+\ell_i,  \ \ i=1,2\}$.

Define the blocks of the product matrices $\mathbf{M}_1 = M_2^{-\top}M_1$  and $\mathbf{M}_2 = M_1^{-\top}M_2$ as follows:
\begin{align*}
\mathbf{M}_1  = 
\begin{bmatrix} 
\mathbf{A}_1 & \mathbf{B}_1 \\ \mathbf{C}_1 & \mathbf{D}_1 
\end{bmatrix}, \ \ \text{and}\ \ 
\mathbf{M}_2 = 
\begin{bmatrix} 
\mathbf{D}_2 & \mathbf{C}_2 \\ \mathbf{B}_2 & \mathbf{A}_2
\end{bmatrix}.
\end{align*}


\begin{theorem}
\label{thm:invupdate}
The composite LFT update in \eqref{eq:update1} can be written in the compact form
\begin{equation}
    \label{eq:main_composite_dynamics}
    L_i^+ = \big(\mathbf{C}_i+\mathbf{D}_iL_i\big)\big(\mathbf{A}_i+\mathbf{B}_iL_i\big)^{-1}.
\end{equation}
\end{theorem}

\begin{proof}
We show the proof for $i=1$ and $-i=2$ for clarity.
Expanding $\mathbf{M}_1=M_2^{-\top}M_1$ by using block matrix inversion on $M_2^{-\top}$
\begin{align*}
M_2^{-\top} = \bmat{
\big(D_2^\top - B_2A_2^{-\top} B_2^\top \big)^{-1}
& - \big(D_2^\top - B_2A_2^{-\top} B_2^\top \big)^{-1} B_2 A_2^{-\top} \\
- A_2^{-\top} B_2^\top \big(D_2^\top - B_2A_2^{-\top} B_2^\top \big)^{-1} & 
A_2^{-\top} + A_2^{-\top} B_2^\top 
\big(D_2^\top - B_2A_2^{-\top} B_2^\top \big)^{-1}
B_2 A_2^{-\top}
}
\end{align*}
we deduce that
\begin{align}
\mathbf{M}_1 = M_2^{-\top }M_1 
= \begin{bmatrix}
G^{-1}E & G^{-1}F  \\
A_2^{-\top}(B_1 -B_2^\top G^{-1} E)
& 
A_2^{-\top}(D_1 
- B_2^\top G^{-1}F)
\end{bmatrix},
\label{eq:M1expand}
\end{align}
%
with 
\begin{align*}
G = D_2^\top - B_2 A_2^{-\top}B_2^\top, \qquad 
E = A_1 - B_2 A_2^{-\top}B_1, \qquad 
F = B_1^\top - B_2 A_2^{-\top} D_1.
\end{align*}
We have specifically chosen a block matrix inversion that requires $A_2^\top$ and 
$G$ to be invertible, yet  does not explicitly require $D_2$ to be invertible ---in many practical cases it will not be. 
Proceeding from the update \eqref{eq:update1}, we have that
\begin{align*}
    L_1^+
&= - \left[A_{2}^\top - (B_1 +D_1L_1)(A_1 + B_1^\top  L_1)^{-1}B_{2}\right]^{-1} 
\Big(B_{2}^\top  - (B_1 +D_1L_1)(A_1 + B_1^\top  L_1)^{-1} D_{2}^\top \Big).
\end{align*}
Applying the Woodbury matrix identity to the inverse 
\begin{align*}
\big[\cdot \big]^{-1} 
& = \left[A_{2}^\top - (B_1 +D_1L_1)(A_1 + B_1^\top  L_1)^{-1}B_{2}\right]^{-1} \\
& =
A_2^{-\top} + A_2^{-\top}
\big(B_1 +D_1L_1\big)\Big(
A_1 + B_1^\top  L_1 
- B_2 A_2^{-\top} \big(B_1 + D_1 L_1 \big)
\Big)^{-1}B_2 A_2^{-\top} \\
& =
A_2^{-\top} + A_2^{-\top}
\big(B_1 +D_1L_1\big)\Big(
E + FL_1
\Big)^{-1}B_2 A_2^{-\top}
\end{align*}
and collecting terms,
we deduce that
\begin{equation*}
   \begin{aligned}
    L_1^+&=- A_{2}^{-\top}
B_{2}^\top  + A_{2}^{-\top} (B_1 +D_1L_1)(A_1 + B_1^\top  L_1)^{-1} D_{2}^\top   \\
&  
- A_{2}^{-\top}
(B_1 +D_1L_1)  ( E + F L_1)^{-1} \Big[B_2 A_2^{-\top}B_{2}^\top  
- B_2 A_2^{-\top}  (B_1 +D_1L_1)(A_1 + B_1^\top  L_1)^{-1} D_{2}^\top \Big]
\end{aligned} 
\end{equation*}
After some algebraic manipulation, we have that the last multiplicative term in $[\cdot]$ satisfies
\begin{equation*}
    \begin{aligned}
    \big[\cdot\big] & = B_2 A_2^{-\top} B_2^\top - B_2 A_2^{-\top} \big(B_1 +D_1L_1\big)\big(A_1 + B_1^\top  L_1\big)^{-1} D_{2}^\top\\
    & =
 -G + D_2^\top 
 +\big(- B_2 A_2^{-\top}B_1 - B_2 A_2^{-\top}D_1 L_1 \big)\big(A_1 + B_1^{\top}  L_1\big)^{-1} D_{2}^\top \\
 & = 
 -G 
 +\big(A_1 + B_1^{\top}  L_1 - B_2 A_2^{-\top}B_1 - B_2 A_2^{-\top}D_1 L_1 \big)\big(A_1 + B_1^{\top}  L_1\big)^{-1} D_{2}^\top \\
& = 
-G + (E + F L_1 ) (A_1 + B_1^\top  L_1)^{-1} D_2^\top.
    \end{aligned}
\end{equation*}
Substituting this into the expression for $L_1^+$ and simplifying gives 
\begin{align*}
    L_1^+&=- A_{2}^{-\top}
B_{2}^\top  + A_{2}^{-\top} (B_1 +D_1L_1)(A_1 + B_1^\top  L_1)^{-1} D_{2}^\top  \\
& \quad 
- A_{2}^{-\top}
(B_1 +D_1L_1)  ( E + F L_1 )^{-1} 
\left(-G + (E + F L_1) (A_1 + B_1^\top  L_1)^{-1} D_2^\top\right). \\
&=- A_{2}^{-\top}
B_{2}^\top  + A_{2}^{-\top} (B_1 +D_1L_1)(A_1 + B_1^\top  L_1)^{-1} D_{2}^\top  \\
& \quad 
+ A_{2}^{-\top}
(B_1 +D_1L_1)  ( E + F L_1 )^{-1} G
-A_{2}^{-\top}
(B_1 +D_1L_1) (A_1 + B_1^\top  L_1)^{-1} D_2^\top \\
&=- A_{2}^{-\top}
B_{2}^\top  
+ A_{2}^{-\top}
(B_1 +D_1L_1)  ( E + F L_1 )^{-1} G
\end{align*}
Arranging further and finally comparing with \eqref{eq:M1expand} gives 
\begin{equation*}
 \begin{aligned}
    L_1^+
&=\Big(A_2^{-\top}\big(B_1+D_1L_1\big)-A_{2}^{-\top}B_2^\top G^{-1}\big(E+FL_1\big)\Big)\big(E+FL_1\big)^{-1}G\\
&=\Big(A_2^{-\top}\big(B_1-B_2^\top G^{-1}E \big) +
A_2^{-\top}\big(D_1-B_2^\top G^{-1}F\big)L_1 \Big)
\Big(G^{-1}E+G^{-1}FL_1\Big)^{-1}\\
&=(\mathbf{C}_1 + \mathbf{D}_1L_1 )(\mathbf{A}_1 + \mathbf{B}_1L_1)^{-1}, 
\end{aligned}   
\end{equation*}
which concludes the proof.
\end{proof}
We note that this update can be written as
\begin{align}
\begin{bmatrix} I \\ L_1^+ \end{bmatrix} 
= 
\mathbf{M}_1
\begin{bmatrix} I \\ L_1 \end{bmatrix}
\Big[ \mathbf{A}_1+\mathbf{B}_1L_1 \Big]^{-1}
\label{eq:invcomp2}
\end{align}
Iterating from an initial conjecture $L_1(0)$ gives 
\begin{align}
\begin{bmatrix} I \\ L_1(1) \end{bmatrix} 
=
\mathbf{M}_1
\begin{bmatrix} I \\ L_1(0) \end{bmatrix}
\Big[ \mathbf{A}_1+\mathbf{B}_1L_1(0) \Big]^{-1} 
\label{eq:iter1}
\end{align}
Iterating again and then plugging in \eqref{eq:iter1} for the $(*)$ term gives 
\begin{align*}
\begin{bmatrix} I \\ L_1(2) \end{bmatrix} 
= 
\mathbf{M}_1
\underbrace{
\begin{bmatrix} I \\ L_1(1) \end{bmatrix}}_{(*)}
\Big[ \mathbf{A}_1+\mathbf{B}_1L_1(1) \Big]^{-1}
= 
\mathbf{M}_1
\cdot 
\mathbf{M}_1
\begin{bmatrix} I \\ L_1(0) \end{bmatrix}
\Big[ \mathbf{A}_1+\mathbf{B}_1L_1(0) \Big]^{-1}
\Big[ \mathbf{A}_1+\mathbf{B}_1L_1(1) \Big]^{-1}
\end{align*}
Continuing the iteration process for $k$ steps leads to
\begin{align}
\begin{bmatrix} I \\ L_1(k) \end{bmatrix} 
& 
=
\begin{bmatrix} 
\mathbf{A}_1 & \mathbf{B}_1 \\ \mathbf{C}_1 & \mathbf{D}_1 
\end{bmatrix}^k
\begin{bmatrix} I \\ L_1(0) \end{bmatrix}\Pi_{t=0}^{k-1}(\mathbf{A}_1+\mathbf{B}_1L_1(t))^{-1}
\end{align}


 In some sense the evolution of $L_1(k)$ is given by repeated application of $\mathbf{M}_1$ as in a discrete time linear system.  However, the right multiplication by $\Pi_{t=0}^{k-1}(\mathbf{A}_1+\mathbf{B}_1L_1(t))^{-1}$
makes the evolution nonlinear and more complicated.  Some features of the evolution of linear systems do apply, however.  Specifically if $[I; L_1(0)]$
initially spans an  $\mathbf{M}_1$--invariant subspace, then $[I; L_1(k)]$
will remain within that subspace as well for all $k$.  This fact is at the heart of the equilibrium analysis in the next section.

\section{Equilibrium Analysis via Invariant Subspaces}

We can find equilibrium points for the LFT dynamics using invariant subspaces.  
The following theorem defines fixed points of the composite LFT dynamics \eqref{eq:main_composite_dynamics} from which CCVE can be directly computed.
\begin{theorem}[Equilibrium Computation]
\label{thm:equilib}
Let 
$K_1 = [ Y_1; X_1] \in \mb{C}^{d\times d_1}$ where $Y_1\in \mb{C}^{d_1\times d_1}$ and $X_1\in \mb{C}^{d_2\times d_1}$ 
define an $\mathbf{M}_1$--invariant subspace where $Y_1$ is square and nonsingular.
It follows that $L_1 = X_1Y_1^{-1}\in \mb{C}^{d_2\times d_1}$ is fixed point of the composite LFT dynamics \eqref{eq:main_composite_dynamics}. A completely analogous statement holds for $L_2=X_2Y_2^{-1}$. 
\end{theorem}

\begin{proof} 
Select the columns of $K_1$ to span a right-invariant subspace of $\mathbf{M}_1$, so that
$M_1^{-1}M_2^\top K_1 = K_1\Lambda$.   
In general, $K_1$ can be complex leading to complex conjectures. For problems with real parameters, however, $K_1$ can often be chosen to be real. Even if the invariant subspace contains conjugate pairs of eigenvectors, $K_1$ can be chosen to be a real basis with vectors spanning any planes of rotation and $\Lambda$ will simply be block diagonal as opposed to diagonal. The one exception to this is if the $\mathbf{M}_1$-invariant subspace contains only one complex eigenvector from a complex conjugate pair (see Remark \ref{rem:distinct} below).
Since $\mathbf{M}_1$ is invertible, the matrix $\Lambda$ will be as well. Hence we have that
\begin{align*}
\mathbf{M}_1
\begin{bmatrix}
Y_1 \\ X_1
\end{bmatrix}
=
\begin{bmatrix}
Y_1 \\ X_1
\end{bmatrix}
\Lambda
\implies
\begin{bmatrix}
\mathbf{A}_1 & \mathbf{B}_1 \\ \mathbf{C}_1 & \mathbf{D}_1 
\end{bmatrix}
\begin{bmatrix}
I \\ L_1 
\end{bmatrix}
=
\begin{bmatrix}
I \\ L_1
\end{bmatrix}
\mathbf{H}_1
\end{align*}
where we have right multiplied by $Y_1^{-1}$ and plugged in $L_1$ and $\mathbf{H}_1 = Y_1\Lambda Y_1^{-1}$. Note that $\mathbf{H}_1$ is invertible.

The top equation gives $\big(\mathbf{A}_1+\mathbf{B}_1L_1) = \mathbf{H}_1$.
Plugging this result into the bottom equation gives $\mathbf{C}_1 + \mathbf{D}_1L_1  = L_1(\mathbf{A}_1+\mathbf{B}_1L_1)$ which implies $L_1 = \big(\mathbf{C}_1 + \mathbf{D}_1L_1)(\mathbf{A}_1+\mathbf{B}_1L_1)^{-1}$.
This verifies that $L_1=X_1Y_{1}^{-1}$ is a fixed point of the dynamics as claimed which completes the proof. 
\end{proof}
We note that in the case where $Y_1$ is not invertible in the construction above, this method cannot be used and we leave analysis of this case to future work.  

While the choice of $\mathbf{M}_1$--invariant subspace matters for the computation of the equilibrium, the choice of basis for this space does not.  
\begin{proposition}[Invariance with respect to basis.]
\label{prop:basis}
Let  $K_1 = \big[  Y_1; \ X_1 \big]$ and $K_1' = \big[  Y_1'; \ X_1'\big]$ be two different bases for the same $\mathbf{M}_1$--invariant subspace with $Y_1,Y_1'$ square and non-singular. Then $L_1 = X_1Y_1^{-1} = X_1'Y_1'^{-1}$.
\end{proposition}
\begin{proof}
Since $K_1$ and $K_1'$ are bases for the same space, there exists square, non-singular $W$ such that $K' = KW$.  It follows that $X_1'Y_1'^{-1} = X_1WW^{-1}Y_1^{-1} = X_1Y_1^{-1}$, which completes the proof.
\end{proof}


\subsection{Alternative Computation}
\label{sec:alt}
The equilibrium solution can be derived from \eqref{eq:update1}  using an alternative method without initially showing that the composite LFT map is given by the formula in Theorem \ref{thm:invupdate}.
Since the analysis is more direct---and also provides inspiration for Theorem \ref{thm:invupdate} and a useful perspective for proofs later on---we reproduce it here.  
Expanding and rearranging \eqref{eq:update1} at equilibrium, we get that $A_2^\top L_1-  (B_1 +D_1L_1)(A_1 + B_1^\top L_1)^{-1} B_2L = - B_2^\top  + (B_1 +D_1L_1)(A_1 + B_1^\top L_1)^{-1}D_2^\top$ which implies 
\begin{equation}\label{eq:rearrange1}
  \begin{aligned}
A_2^\top  L_1+B_2^\top  &= (B_1 +D_1L_1)(A_1 + B_1^\top  L_1)^{-1}
 (D_2^\top  + B_2L_1). 
\end{aligned}  
\end{equation}
Using this form of the fixed point equations, we can solve for the equilibrium using a similar invariant subspace argument. 
\begin{proposition}[Alternative Equilibrium Computation]
\label{prop:alt}
Let the columns of $K_1 = \big[  Y_1^\top \ X_1^\top \big]^\top$ solve the generalized eigenvalue problem $M_1K_1 = M_2^\top K_1 \Lambda$. Then $L_1 = X_1Y_1^{-1}$ solves \eqref{eq:rearrange1}.  
\end{proposition}
\begin{proof} The expression 
$M_1K_1  = M_2^\top K_1 \Lambda$ gives 
\begin{align}
\begin{bmatrix}
A_1 & B_1^\top  \\
B_1 & D_1 
\end{bmatrix}
\begin{bmatrix}
I \\ L_1
\end{bmatrix}
= 
\begin{bmatrix}
D_2^\top  & B_2 \\
B_2^\top  & A_2^\top  
\end{bmatrix}
\begin{bmatrix}
I \\ L_1
\end{bmatrix} \mathbf{H}_1
\label{eq:blab}
\end{align}
where $\mathbf{H}_1= Y_1 \Lambda Y_1^{-1}$. This expression arises since we have right multiplied by $Y_1^{-1}$ and plugged in $L_1=X_1Y_1^{-1}$.  Again, since $\mathbf{M}_1$ is non-singular, $\mathbf{H}_1$ will be as well.
The top and bottom equation, respectively, can be rearranged to deduce that {\small $( A_1+B_1^\top L_1)^{-1}( D_2^\top+B_2 L_1)  = \mathbf{H}_1^{-1}$} so that
$(B_1 + D_1 L_1)\mathbf{H}_1^{-1}  =(B_2^\top + A_2^\top L_1 )$.
Plugging in $\mathbf{H}_1^{-1}$ leads to \eqref{eq:rearrange1}, which concludes the proof.
\end{proof}

Inspiration for the the composite dynamics can then be seen by noting that for invertible $M_2$,  we see that 
$M_1K_1 = M_2^\top K_1 \Lambda
\ \iff \ M_2^{-\top} M_1K_1 = K_1 \Lambda.$

At first pass, there are many ways to choose an $\mathbf{M}_1$--invariant subspace to compute $L_1$.  Explicitly, there are 
$d$ choose $d_1$ 
ways to select a basis of eigenvectors.  A further stability analysis, however, shows that there is only one way to select an invariant subspace that leads to a stable $L_1$ when the eigenvalues of $\mathbf{M}_1$ have distinct magnitudes.  This analysis is given in Section \ref{sec:stability}.

\section{Equilibrium Stability}
\label{sec:stability}
We next characterize the stability properties of fixed points of \eqref{eq:coupledriccati}---which includes the set of CCVE---and show how stability is related to the matrices $\mathbf{M}_i$, $i=1,2$.

The local stability  of a nonlinear system can be characterized by examining the eigenstructure of the local linearization; in particular, by the Hartman-Grobman theorem, if the eigenvalues of the local linearization evaluated at a fixed point of the nonlinear system have modulus less than one, then the fixed point is a locally asymptotically stable equilibrium of the nonlinear system.
\begin{theorem}[Perturbation Dynamics]
The linearized perturbation dynamics at equilibrium are
$
\Delta L_i^+   = \mathbf{\Omega}_i(\Delta L_i)  = (\mathbf{D}_i  -  L_i \mathbf{B}_i) \Delta L_i (\mathbf{A}_i + \mathbf{B}_iL_i)^{-1}$. 
\end{theorem}
\begin{proof}
 Perturbing the equilibrium conjectures gives $L_i^+ + \Delta L_i^+= (\mathbf{C}_i + \mathbf{D}_i L_i + \mathbf{D}_i \Delta L_i )
 (\mathbf{A}_i + \mathbf{B}_iL_i+\mathbf{B}_i\Delta L_i)^{-1}$.
 At equilibrium $L_i = L_i^+$, we have
 that $(L_i + \Delta L_i^+)
(\mathbf{A}_i + \mathbf{B}_iL_i+\mathbf{B}_i \Delta L_i)=(\mathbf{C}_i + \mathbf{D}_i L_i + \mathbf{D}_i \Delta L_i )$.
Recall that in equilibrium $L_i(\mathbf{A}_i + \mathbf{B}_iL_i) - (\mathbf{C}_i + \mathbf{D}_i L_i) = 0$.
Therefore,
 we deduce that 
  $\Delta L_i^+  =  (\mathbf{D}_i  -  L_i \mathbf{B}_i ) \Delta L_i 
(\mathbf{A}_i+ \mathbf{B}_i L_i +\mathbf{B}_i \Delta L_i  )^{-1}$.
Applying the Woodbury matrix identity to the inverse and noting limits we further deduce that
 \begin{align*}
\Delta L_i^+  & =  \big(\mathbf{D}_i  -  L_i \mathbf{B}_i ) \Delta L_i (\mathbf{A}_i + \mathbf{B}_iL_i)^{-1}  - 
 (\mathbf{D}_i  -  L_i \mathbf{B}_i ) \Delta L_i 
(\mathbf{A}_i + \mathbf{B}_iL_i)^{-1}\mathbf{B}_i \\
&\qquad \cdot 
\left(
I + \Delta L_i (\mathbf{A}_i + \mathbf{B}_iL_i)^{-1}\mathbf{B}_i \right)^{-1}\Delta L_i(\mathbf{A}_i + \mathbf{B}_iL_i)^{-1}.
 \end{align*}
Noting that 
$
\left(
I + \Delta L_i (\mathbf{A}_i + \mathbf{B}_iL_i)^{-1}\mathbf{B}_i \right)^{-1} \to I
$ 
as $\Delta L_i \to 0$ and then dropping higher order terms completes the proof.
\end{proof}
Note that $\mathbf{\Omega}_i(\cdot)$ for $i=1,2$ are linear operators in the form of a discrete time Lyapunov equation.  To understand their stability, we recall a result from discrete time Lyapunov theory given here without proof.

\begin{lemma}[DT Lyapunov Operators]
\label{lem:dtlyap2}
For $A,B \in \mathbb{C}^{n \times n}$, the linear operator $\mathcal{A}(X) = AXB$ 
has eigenvalues of the form $\lambda_j\mu_k$ where $\lambda_j \in \text{spec}(A)$ and $\mu_k \in \text{spec}(B)$.  
\end{lemma}
The following characterization of the spectra of $\mathbf{\Omega}_i(\cdot)$ then follows immediately. 
\begin{theorem}
\label{specratio}
The spectrum of the linear operator $\mathbf{\Omega}_i(\cdot)$ is given by
 \[\spec(\mathbf{\Omega}_i ) = \left\{ \frac{\lambda_j}{\mu_k} \big| \lambda_j \in \spec(\mathbf{D}_i-L_i\mathbf{B}_i), \ \mu_k\in \spec(\mathbf{A}_i+\mathbf{B}_iL_i)\right\}.\] 
\end{theorem}

The next theorem establishes equivalent conditions for local stability. 
\begin{theorem}
\label{thm:uniquestable}
Given a fixed point $(L_1^{\tt c},L_2^{\tt c})$ of \eqref{eq:update1}, without loss of generality, the following are equivalent statements:
\begin{enumerate}
    \item[a.]  The fixed point $(L_1^{\tt c},L_2^{\tt c})$ is locally asymptotically stable with respect to \eqref{eq:update1} for $i=1,2$;
    \item[b.] The eigenvalues $\xi_j\in \spec(\mathbf{\Omega}_1(L_1^{\tt c}))$ as such that $|\xi_j|<1$ for all $j$.
    \item[c.]  The matrix $K_1\in \mb{C}^{d\times d_1}$ from Theorem \ref{thm:equilib} (and Proposition \ref{prop:alt}) is chosen to span an  $\mathbf{M}_1$--invariant subspace constructed using the $d_1$ largest magnitude eigenvalues;
\end{enumerate}
\end{theorem}

Theorem~\ref{thm:uniquestable} not only establishes equivalent conditions for stability, but also shows that it is sufficient to establish stability for one player in order to show the  combined dynamics (i.e., \eqref{eq:update1} for $i=1,2$) are stable. However, the (local) rates of convergence for each player will depend on the eigenstructure of their individual dynamics.
\begin{corollary}
Players locally converge to $(L_1^{\tt c},L_2^{\tt c})$ with iteration complexity $O(\xi_{i,\max}^k)$ where $\xi_{i,\max}:=\max_{\xi\in \spec(\Omega_i(L_i^{\tt c})}{|\xi|}$ for player $i=1,2$, respectively.
\end{corollary}

This result is analogous to the scalar M\"obius transformation case and to prove it we need to further elucidate the eigenstructure of $\mathbf{M}_i$ which we do in the next section.


The following proposition (which is of independent interest) characterizes the eigenstructure of $\mathbf{M}_1$, without loss of generality, and is used to prove Theorem~\ref{thm:uniquestable}.
\begin{proposition}
\label{thm:sim}
The matrices $L_1$ computed from Theorem \ref{thm:invupdate} and $L_2$ from \eqref{eq:L2fromL1} define the following similarity transforms on $\mathbf{M}_1$
and  $\mathbf{M}_2$, respectively:
\begin{align}
\begin{bmatrix}
I & 0 \\-L_1 &  I
\end{bmatrix}
\begin{bmatrix}
\mathbf{A}_1 & \mathbf{B}_1\\ 
\mathbf{C}_1 & \mathbf{D}_1
\end{bmatrix}
\begin{bmatrix}
I & 0 \\ L_1 &  I
\end{bmatrix}
&
= 
\begin{bmatrix} 
\mathbf{H}_1 & \mathbf{B}_1 \\
0 & \mathbf{H}_1'
\end{bmatrix} 
\label{eq:triangle1} \\
\begin{bmatrix}
I & -L_2\\ 0 & I  
\end{bmatrix}
\begin{bmatrix}
\mathbf{D}_2 & \mathbf{C}_2\\ 
\mathbf{B}_2 & \mathbf{A}_2
\end{bmatrix}
\begin{bmatrix}
I & L_2 \\ 0 & I  
\end{bmatrix} 
& =
\begin{bmatrix} 
\mathbf{H}_2' & 0 \\
 \mathbf{B}_2 &  \mathbf{H}_2
\end{bmatrix}
\label{eq:triangle2}
\end{align}
where $
\mathbf{H}_1 = \mathbf{A}_1 + \mathbf{B}_1L_1 
$, $\mathbf{H}_1' = \mathbf{D}_1 -L_1 \mathbf{B}_1$,
$\mathbf{H}_2  = \mathbf{A}_2 + \mathbf{B}_2 L_2 $, and 
$\mathbf{H}_2' = \mathbf{D}_2 - L_2 \mathbf{B}_2 $.  
Furthermore, the spectrum of the $\mathbf{M}_1$-invariant subspace spanned by $[I; L_1]$
is $\text{spec}(\mathbf{H}_1)$ and the spectrum of the 
$\mathbf{M}_2$-invariant subspace spanned by $[L_2; I]$ is $\text{spec}(\mathbf{H}_2)$ and we can also write 
\begin{align*}
\mathbf{H}_1 & = \mathbf{A}_1 + \mathbf{B}_1L_1 = 
\big(D_2^\top + B_2 L_1 \big)^{-1}
\big(A_1 + B_1^\top L_1 \big)  \\
\mathbf{H}_2' & = \mathbf{D}_2 - L_2 \mathbf{B}_2 = 
(A_1 + B_1^\top L_1)^{-\top}(D_2^\top + B_2L_1)^\top.
\end{align*}
and $\mathbf{H}_2'$ is similar to $\mathbf{H}_1^{-\top}$.

\end{proposition}
\begin{proof}
The expression in \eqref{eq:triangle1} is immediate with the zero block coming from the fixed point equation \eqref{eq:main_composite_dynamics}:
\begin{align*}
\begin{bmatrix}
I & 0 \\-L_1 &  I
\end{bmatrix}
\begin{bmatrix}
\mathbf{A}_1 & \mathbf{B}_1\\ 
\mathbf{C}_1 & \mathbf{D}_1
\end{bmatrix}
\begin{bmatrix}
I & 0 \\ L_1 &  I
\end{bmatrix}
=
\begin{bmatrix} 
\mathbf{A}_1  + \mathbf{B}_1L_1 & \mathbf{B}_1 \\
-L_1 (\mathbf{A}_1 + \mathbf{B}_1 L_1) + 
\mathbf{C}_1 + \mathbf{D}_1 L_1 
& \mathbf{D}_1 - L_1 \mathbf{B}_1
\end{bmatrix} 
=
\begin{bmatrix} 
\mathbf{H}_1 & \mathbf{B}_1 \\
0 & \mathbf{H}_1'
\end{bmatrix}.
\end{align*}
To see the similarity transform on $\mathbf{M}_2$, note that \eqref{eq:L2fromL1} can be rewritten as 
\begin{align}
\begin{bmatrix}
I \\ -L_2^\top 
\end{bmatrix}
= 
M_1
\begin{bmatrix}
I \\ L_1
\end{bmatrix}
(A_1 + B_1^\top L_1)^{-1}.
\label{eq:asdf1}
\end{align}
Transposing and expanding allows us to write 
\begin{align*}
\begin{bmatrix} I \  -L_2 \end{bmatrix} 
\mathbf{M}_2
& = \begin{bmatrix} I \  -L_2 \end{bmatrix} M_1^{-\top} M_2   = (A_1 + B_1^\top L_1)^{-\top} \begin{bmatrix} I \ \ L_1^\top \end{bmatrix} M_2.
\end{align*}
The expression in \eqref{eq:blab}
gives 
\[M_2^\top \bmat{I \\ L_1}
=M_1 \bmat{I \\ L_1}\mathbf{H}_1^{-1},\]
which allows us to write 
\begin{align*}
\begin{bmatrix} I \  -L_2 \end{bmatrix} \mathbf{M}_2
& = (A_1 + B_1^\top L_1)^{-\top} \mathbf{H}_1^{-\top} \begin{bmatrix} I \  L_1^\top \end{bmatrix} M_1^\top \\
& = (A_1 + B_1^\top L_1)^{-\top} \mathbf{H}_1^{-\top} 
(A_1 + B_1^\top L_1)^{\top}
(A_1 + B_1^\top L_1)^{-\top}
\begin{bmatrix} I \  L_1^\top \end{bmatrix} M_1^\top \\
& = \mathbf{H}_2' \begin{bmatrix} I \ -L_2 \end{bmatrix}
\end{align*}
where in the last step we have used \eqref{eq:asdf1} again and substituted 
$
\mathbf{H}_2'
 = [A_1 + B_1^\top L_1]^{-\top} \mathbf{H}_1^{-\top} [A_1 + B_1^\top L_1]^{\top}  $
Note that $\mathbf{H}_2'$ and $\mathbf{H}_1^{-\top}$ are similar. 
From the above expression, we deduce the top row of the following equation:
\begin{align*}
\begin{bmatrix} I & -L_2 \\ 0 & I  \end{bmatrix}
\begin{bmatrix}
\mathbf{D}_2 & \mathbf{C}_2 \\
\mathbf{B}_2 & \mathbf{A}_2
\end{bmatrix}
=
\begin{bmatrix} 
\mathbf{H}_2' & 0 \\
 \mathbf{B}_2 &  \mathbf{A}_2 + \mathbf{B}_2 L_2
\end{bmatrix}
\begin{bmatrix} I & -L_2 \\ 0 & I \end{bmatrix}
\end{align*}
and the bottom row is then immediate. Right multiplying by $\left[\substack{I \ L_2\\ 0 \ I}\right] $ gives \eqref{eq:triangle2}.
The characterization of the invariant subspaces spanned by $[I; L_1]$ and $[L_2; I]$ follows immediately from the block diagonal structure.  The alternate characterizations of $\mathbf{H}_1$ and $\mathbf{H}_2'$ follow from the characterization of $\mathbf{H}_1$ given in Proposition \ref{prop:alt} and the definition of $\mathbf{H}_2'$ above which concludes the proof. \end{proof}

We now prove Theorem \ref{thm:uniquestable}.
\begin{proof}[Proof of Theorem \ref{thm:uniquestable}] 
The fact that a.$\Longleftrightarrow$b.~is immediate from Hartman-Grobman \cite{sastry2013nonlinear}. Hence, it only remains to show that a.$\Longleftrightarrow$c. 

The block diagonal structure given in Theorem \ref{thm:sim} gives that 
$\text{spec}(\mathbf{M}_i) = \text{spec}(\mathbf{H}_i) \sqcup \text{spec}(\mathbf{H}_{i}')$. Since the perturbation dynamics are given by $\Delta L_i^+ = \mathbf{H}_i' \Delta L_i \mathbf{H}_i^{-1}$ the result immediately follows for $L_1$ noting that $\spec\big(\mathbf{H}_1\big)$ corresponds to the invariant subspace $K_1$ and the Lyapunov stability arguments in Lemma \ref{lem:dtlyap2}.  Since $\mathbf{H}_1^{-\top}$ and $\mathbf{H}_2'$ are similar (see Theorem~\ref{thm:sim}) and noting that $\spec(\mathbf{M}_1)=1/\spec(\mathbf{M}_2)$ as stated previously, this choice also implies that player 2's perturbation dynamics are asymptotically stable which in turn implies that the dynamics \eqref{eq:update1} are locally asymptotically stable around $(L_1^{\tt c},L_2^{\tt c})$. \end{proof}

 \begin{remark}
 \label{rem:distinct}
A result of the analysis in Theorem \ref{thm:uniquestable} is that if the eigenvalues of $\mathbf{M}_1$ (and $\mathbf{M}_2$) clearly divide into large and small magnitude sets (of the appropriate number) where all the eigenvalues in the large set are strictly large than those in the small set, then there is a unique way to choose an (asymptotically) stable CCVE.  When the eigenvalues cannot clearly be divided this way, there may be multiple ways to construct (marginally) stable CCVE. Two particularly interesting cases are when there are eigenvalues from the same Jordan subspace or complex eigenvalues from the same conjugate pair in each set. 
In this second case in particular, the only (marginally) stable conjectures will be complex and any associated real conjectures will exhibit oscillatory behavior analogous to elliptic M\"obius transformations. These interesting cases will be examined in future work.  
\end{remark}
\section{Observations on the Second Order Conditions}
\label{sec:2ndorder}
Given the stability considerations given above there is generally a unique (or limited number) of stable first-order CCVE---i.e., first-order CCVE that can be reached via the response dynamics between the players.  Once these stable equilibria are determined, one should check the second order conditions 
\begin{align}
A_i + L_i^\top B_i + B_i^\top L_i + L_i^\top D_i L_i  
\succ 0     
\qquad \qquad \text{for  $i=1,2$.}
\label{eq:2ndorderremark}
\end{align}
to see if the stable first-order CCVE is well-posed---i.e., if the stable point is a CCVE for the game.  In general this is not guaranteed and will depend on the relative magnitudes of the parameters $A_i,B_i,$ and $D_i$.  In particular, a large enough $A_i \succ 0$ will clearly make it more likely for \eqref{eq:2ndorderremark} to be satisfied. Note also that \eqref{eq:2ndorderremark} can be rewritten as 
\begin{align*}
\begin{bmatrix}
I \\ L_1 
\end{bmatrix}^\top 
\begin{bmatrix}
A_1 & B_1^\top \\
B_1 & D_1
\end{bmatrix}
\begin{bmatrix}
I \\ L_1 
\end{bmatrix}
\succ 0, \qquad \qquad 
\begin{bmatrix}
L_2 \\ I
\end{bmatrix}^\top 
\begin{bmatrix}
D_2 & B_2 \\
B_2^\top  & A_2
\end{bmatrix} 
\begin{bmatrix}
L_2 \\ I
\end{bmatrix}
\succ 0, 
\end{align*}
and thus $M_1, M_2 \succ 0 $ 
are sufficient conditions for \eqref{eq:2ndorderremark} to be satisfied; however, in many practical problems, $M_1,M_2 \nsucc 0$ since $D_1$ and $D_2$ may be zero, low rank, indefinite, or even negative definite.  Simple numerical experiments also show that $M_1,M_2 \succ 0$ is far too conservative of a condition and that \eqref{eq:2ndorderremark} often holds even when it does not. Further analysis of this condition is left to future work. 





\begin{figure}[t]
\centering
\subfloat[][Player 1's action space]{
\includegraphics[width=.48\textwidth]{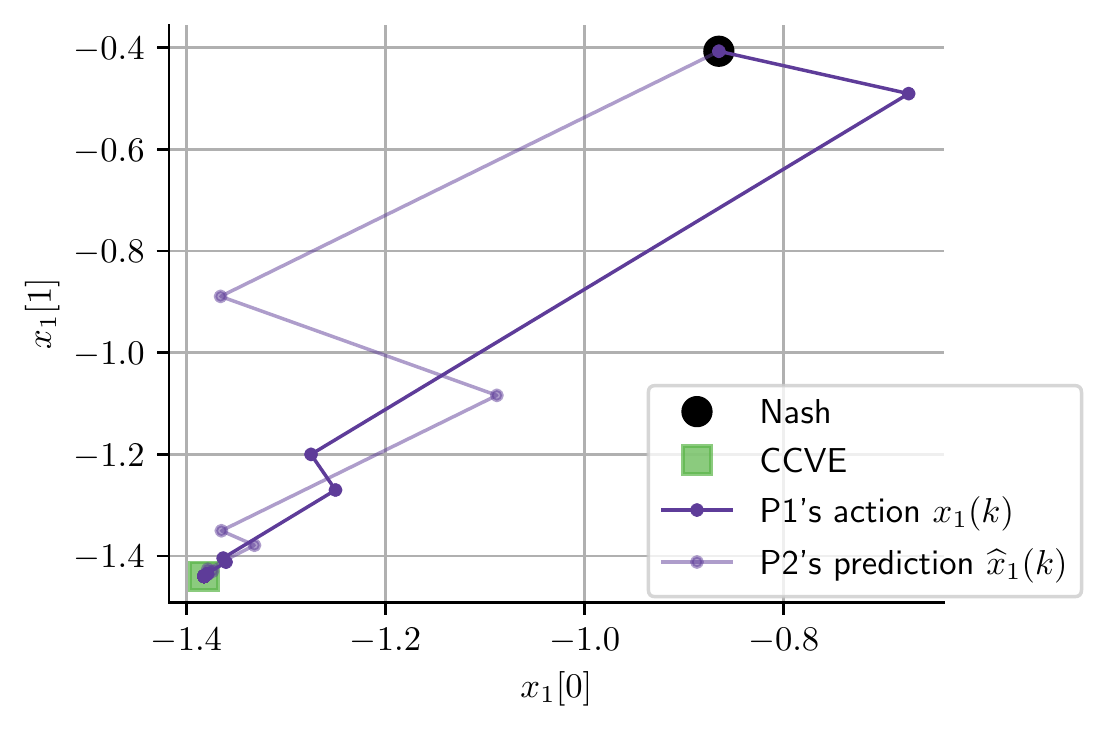}
}
\subfloat[][Player 2's action space]{
\includegraphics[width=.48\textwidth]{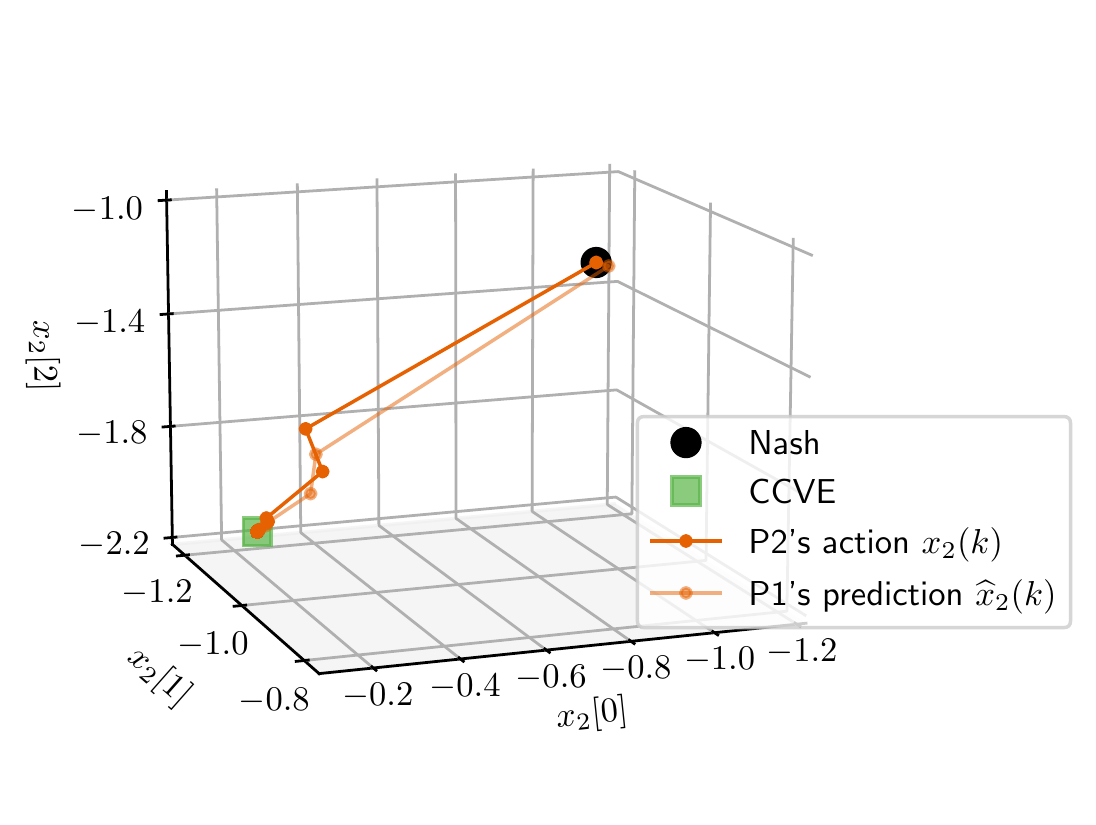}
}
\caption{
\emph{Player Actions and Predictions:} 
The actions and predictions of (a) player 1 and (b) player 2.
Players initialize at the Nash equilibrium, ie. $(L_1,\ell_1) = (0,x_2^{\text{NE}})$ and $(L_2,\ell_2) = (0,x_1^{\text{NE}})$, 
and then update their conjectures via the conjectural variations iteration~\eqref{eq:update1}. 
Each player's action and also their guess of the other player's action are plotted in $\mathbb{R}^2$ and $\mathbb{R}^3$ respectively.  Both the actions and predictions converge to the stable CCVE, however, note that (unsurprisingly) the predictions are not accurate until convergence.}
\label{fig:2x3a}
\end{figure}

\section{Numerics}

In this section, two examples demonstrate the convergence of the proposed iterative method for
quadratic games with different dimensions. The iterative method converges to the unique stable CCVE in both examples.
The code and parameters for the examples are provided\footnote{
The code and parameters for all experiments are available at 
\mbox{\url{https://github.com/bchasnov/2023lcssCCVE}}.
}.

In the first example, we tested the conjectural variations iteration update~\eqref{eq:update1} for a two-player quadratic game on a five-dimensional joint action space,
\mbox{$\mb{R}^{2}\times\R^{3}$}.
The cost matrices $A_i,D_i$ were chosen to be scaled identities, $B_i$ were randomly sampled from a uniform distribution, and $a_i,b_i$ were chosen to be zero or one vectors.
The cost matrices for player 1 were
\begin{align*}
A_1&=\bmat{1 & 0 \\ 0 & 1},
B_1=\bmat{-0.1 & 0.2 \\ -0.5 & -0.2 \\ -0.4 & -0.4},
D_1=-\bmat{0.2 & 0 & 0 \\ 0 & 0.2 & 0 \\  0 & 0 & 0.2 },
a_1=\bmat{0 \\ 0},
b_1=\bmat{0 \\ 0 \\ 0},
\end{align*}
and for player 2 were 
\begin{align*}
A_2&=\bmat{1 & 0 & 0 \\ 0 & 1 & 0 \\ 0  & 0 & 1},
B_2=\bmat{0.3 & 0.2 & 0.1 \\ 
0.0 & 0.1 & -0.2},
D_2=-\bmat{0.1 & 0 \\ 0 & 0.1},
a_2=\bmat{1 \\ 1 \\ 1},
b_2=\bmat{1 \\ 1}.
\end{align*}
At iteration $k$, the players formed affine conjectures $c_{-i}(x_i)=L_i(k)x_i + \ell_i(k),\ i=1,2$ where $L_1(k)\in\R^{3\times 2}$, $L_2(k)\in\R^{2\times 3}$, $\ell_1(k)\in\R^3$ and $\ell_2(k)\in\R^2$.
The players initially play a Nash equilibrium $(x_1^{\text{NE}},x_2^{\text{NE}})$ which is equivalent to assuming conjectures $(L_1,\ell_1) = (0,x_2^{\text{NE}})$ and $(L_2,\ell_2) = (0,x_1^{\text{NE}})$.  
%
For $20$ iterations, each player updates their conjecture via  $L_{1}(k+1) = \LFT_{1,2}( L_2(k) )$ and $L_{2}(k+1) = \LFT_{2,1}( L_1(k) )$ respectively.
At each iteration, player $i$ also computes their optimal action (assuming their current conjecture is correct) by solving $$ 
x_{i}(k) =
\amin_{x_i}\{f_i(x_i,x_{-i})|\ x_{-i}=L_i(k) x_i + \ell_i(k) \}
$$
and their current guess of what the other player is doing $\hat{x}_{-i}(k) = L_i(k)x_{i} (k) + \ell_i(k)$.  Actions for each player as well as the other player's guess of their actions are plotted in Fig.~\ref{fig:2x3a}.  

This iteration process converges to the stable CCVE, confirming the theoretical prediction (Fig.~\ref{fig:2x3b}).  Note that while the players initially play Nash strategies, conjecturing about the other player's optimization problem leads them away from the Nash equilibrium to the CCVE.  This demonstrates that players that make conjectures can end up at dramatically different equilibria than players that do not. 
Note also that each player's prediction of the other player's actions is inaccurate until the conjectures converge and player's reach the CCVE.  Finally, we plot the cost for each player $f_i(x_i(k),x_{-i}(k))$ at each iteration and the total or \emph{social cost} at each iteration 
$$f_{s}\Big(x_1(k),x_2(k)\Big) = f_1\Big(x_1(k),x_2(k)\Big) + f_2\Big(x_1(k),x_2(k)\Big)$$
as well as the optimal social cost $f^*_s = \min_{x_1,x_2} f_{s}(x_1,x_2)$ in Fig.~\ref{fig:2x3c}.  Interestingly in this case, each player's cost is actually lower at the CCVE then at Nash (though this is not guaranteed in all cases).

\begin{figure}[t] 
\centering
\subfloat[][Player 1's action and Player 2's prediction]{
\includegraphics[width=.4\textwidth]{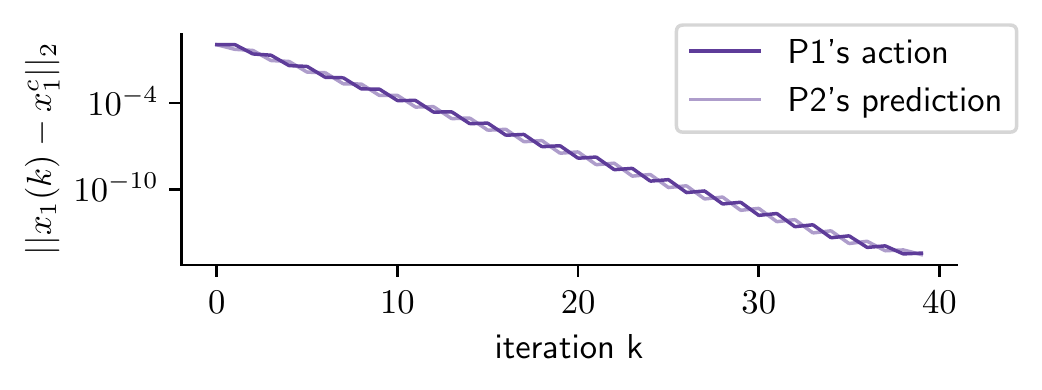}
}
\subfloat[][Player 2's action and Player 1's prediction]{
\includegraphics[width=.4\textwidth]{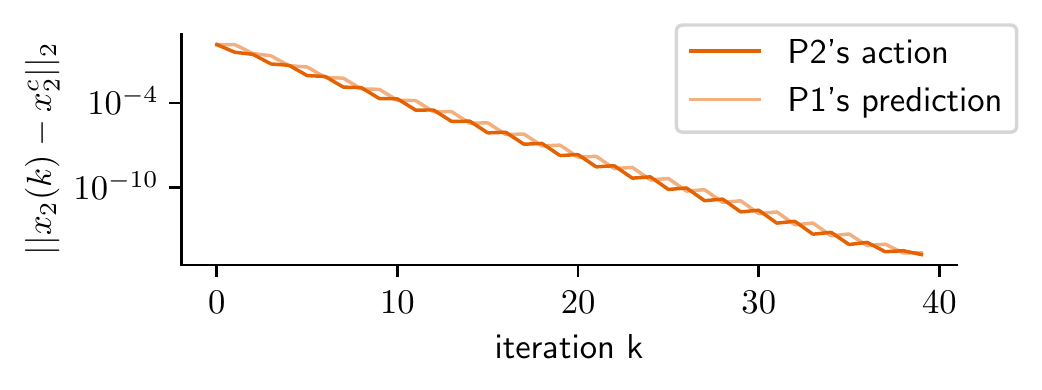}
}

\subfloat[][Player 1's conjecture]{
\includegraphics[width=.4\textwidth]{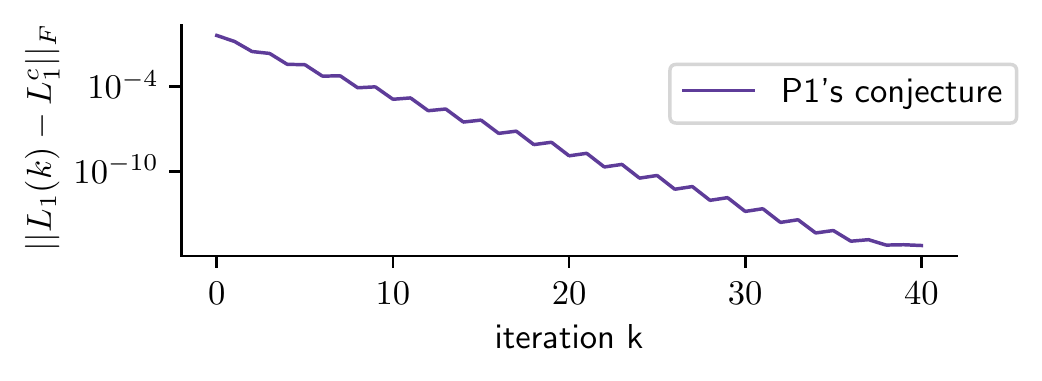}
}
\subfloat[][Player 2's conjecture]{
\includegraphics[width=.4\textwidth]{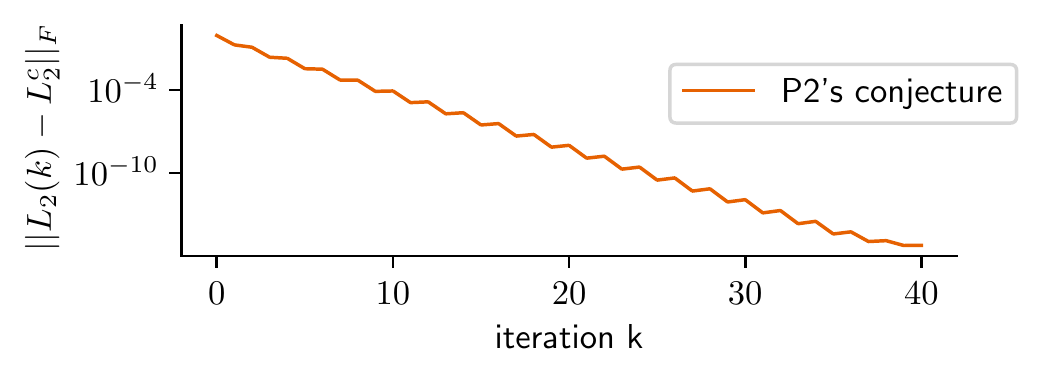}
}
\caption{
\emph{Convergence of Conjectures and Actions:}
The conjectures $(L_1(k),L_2(k))$
converge to the stable equilibrium conjectures 
and both the actions $(x_1(k),x_2(k))$ and predictions $(\hat{x}_1(k),\hat{x}_2(k))$
converge to the CCVE shown here for player 1 (a) and player 2 (b). }
\label{fig:2x3b}
\end{figure}

\begin{figure}[t]
\centering
\subfloat[][Individual costs]{
\includegraphics[width=.4\textwidth]{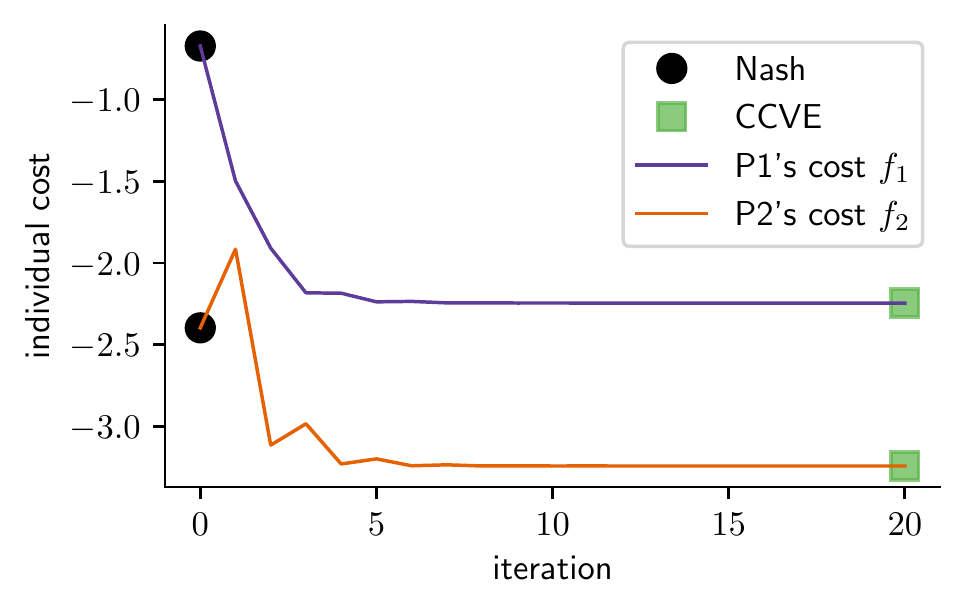}
}
\subfloat[][Social cost]{
\includegraphics[width=.4\textwidth]{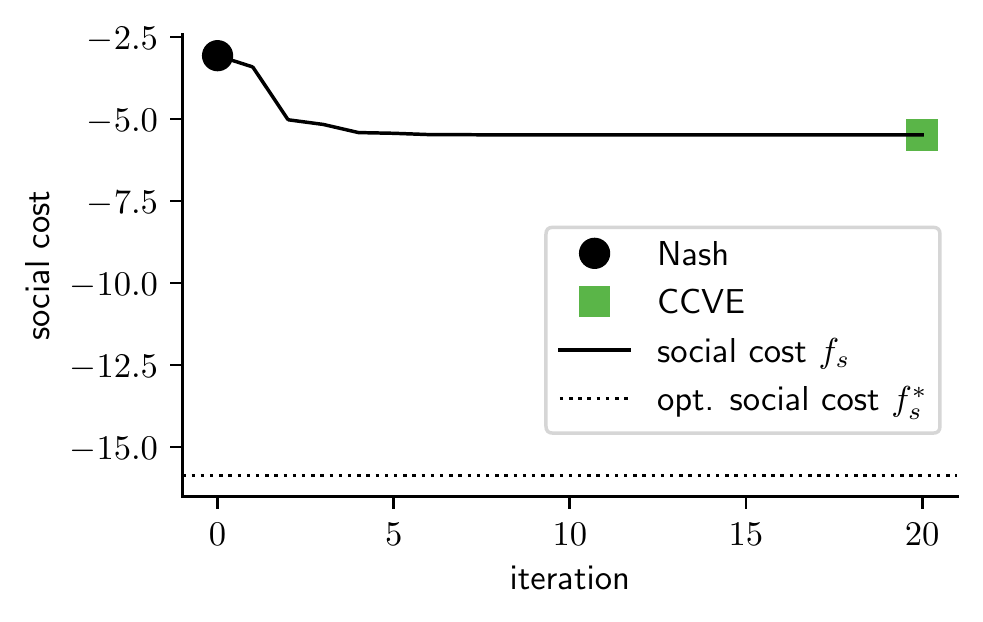}
}
\caption{
\emph{Player and Social Costs:}
The costs for each player $f_i(x_i(k),x_{-i}(k))$ at each iteration are shown in (a) and the social cost $f_s(x_i(k),x_{-i}(k))$ is shown in (b).  The optimal social cost $f_s^*$ is also shown.  Note that the initial costs are the Nash costs and the final costs are the CCVE costs.  Note that in this case the CCVE costs are lower than Nash (though this is not guaranteed in general).}
\label{fig:2x3c}
\end{figure}

%

In the second example, we tested the conjectural variations iteration update~\eqref{eq:update1} for larger two-player quadratic game on 
\mbox{$\R^{50}\times\R^{60}$}. 
The cost matrices for the player 1 were 
$A_1=13I_{d_1},\ D_1=-0.2 I_{d_2},\ a_1=0,\ b_1=0$
and for player 2 were
$A_2=13I_{d_2},\ D_2=- 0.1 I_{d_1},\ a_2=1,\ b_2=1$
where $d_1=50,d_2=60$, the identity matrix $I_{d}$ is $d$ dimensional and $0,1$ are zero and one vectors. Additionally, the cost matrix $B_1,B_2$ was randomly sampled from a uniform distribution between $-1$ and $1$.  
For five randomly sampled quadratic games, we found that the conjectural variation iteration converges to the unique stable CCVE point. The linear fractional transformation was implemented by using standard numerical methods for solving linear matrix equations.
Fig.~\ref{fig:202x303} plots the distance of the iterates from the stable equilibrium decreasing at an exponential rate for all five quadratic games. 
These examples show the computability of the stable consistent equilibria by numerical methods.

\begin{figure}[t]
\centering
\includegraphics[width=.45\textwidth]{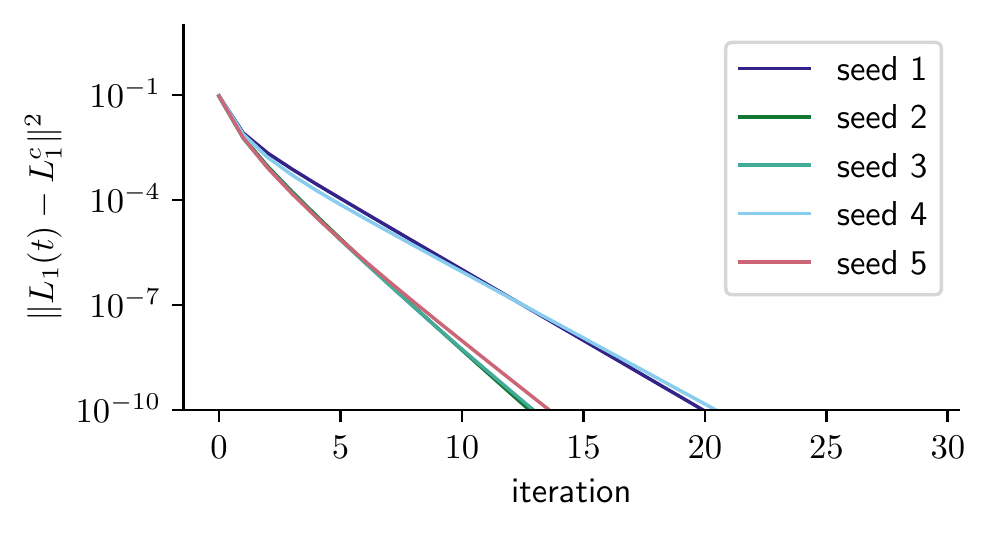}
\caption{
The iteration is demonstrated for quadratic games on $\R^{50}\times\R^{60}$ with randomly sampled cost matrices. 
The distance from the iterate to the stable CCVE is plotted, showing convergence towards the stable CCVE.
}
\label{fig:202x303}
\end{figure}

In both examples,  for comparison, we also compute  
the stable consistent conjectural variations equilibrium directly from the cost matrices without running the conjectural variation iteration. 
We applied Theorem~\ref{thm:equilib} to solve for a fixed point of the composite linear fractional transformation in~\eqref{eq:main_composite_dynamics}.
Selecting the eigenvectors corresponding to the eigenvalues with smallest magnitude to be the columns of $Y,X$, the stable and consistent conjectural variation is $L_1^* = XY^{-1}$.
To verify that $L_1^*$ is an asymptotically stable equilibrium of the conjectural variations iteration, we applied Theorem~\ref{specratio}
and determined that the spectrum of the linear operator of the perturbation dynamics has eigenvalues with magnitude less than one, thus implying that the fixed point $L_1^*$ is locally asymptotically stable with respect to~\eqref{eq:update1}. 
Local rates of convergence will depend on the eigenstructure of the dynamics.






\section{Discussion \& Open Questions}
We introduced a novel analysis of CCVE by drawing on tools from the analysis of coupled Riccati equations. There are a number of interesting open questions including how players might adapt their conjectural variations in both repeated and dynamic games by repeatedly interacting with their opponents, as well as how players might adopt policy gradient like procedures to learn their policies contingent on conjectures adapted over time. 

\bibliographystyle{plainnat}
\bibliography{2023lcsscdc}

\end{document}